\documentclass[journal]{IEEEtran}

\usepackage{algorithmic}
\usepackage{algorithm}
\usepackage{array}
\usepackage[caption=false,font=normalsize,labelfont=sf,textfont=sf]{subfig}
\usepackage{textcomp}
\usepackage{stfloats}
\usepackage{url}
\usepackage{verbatim}
\usepackage{graphicx}
\usepackage{standalone}
\usepackage{hyperref}
\usepackage{amsmath}
\usepackage{amsfonts}
\usepackage{amsthm}
\usepackage{amssymb}
\usepackage{booktabs}
\usepackage{chngpage}
\usepackage{cite}
\usepackage[capitalize]{cleveref}
\usepackage[dvipsnames]{xcolor}
\usepackage{orcidlink}

\usepackage{thmtools, thm-restate}

\title{
Scale-Equivariant Imaging: Self-Supervised Learning for Image Super-Resolution and Deblurring
}

\author{
J\'er\'emy Scanvic \orcidlink{https://orcid.org/0009-0003-6117-0492},
Mike Davies \orcidlink{https://orcid.org/0000-0003-2327-236X},
Patrice Abry \orcidlink{https://orcid.org/0000-0002-7096-8290},
Juli\'an Tachella \orcidlink{https://orcid.org/0000-0003-3878-9142}%
\thanks{J\'er\'emy Scanvic, Patrice Abry and Juli\'an Tachella are with ENSL, CNRS, Laboratoire de Physique, F-69364 Lyon, France.}%
\thanks{Mike Davies is with the School of Engineering, University of Edinburgh, Edinburgh, EH9 3FB, UK.}%
\thanks{This project was provided with computing HPC and storage resources by GENCI at IDRIS thanks to the grant 2025-AD011016422 on the supercomputer Jean Zay's H100 partition.},
}

\IEEEpubid{}
\newtheorem{definition}{Definition}

\newcommand{\eg}{e.g.,}

\newcommand{\csignalset}{\mathcal{Z}}
\newcommand{\dsignalset}{\mathcal{X}}

\newcommand{\varkernel}{{h}}
\newcommand{\varpos}{{u}}
\newcommand{\varposalt}{{v}}

\begin{document}

\maketitle

\begin{abstract}
Self-supervised methods have recently proved to be nearly as effective as
supervised ones in various imaging inverse problems, paving the way for
learning-based approaches in scientific and medical imaging applications where
ground truth data is hard or expensive to obtain. These methods critically rely
on invariance to translations and/or rotations of the image distribution to
learn from incomplete measurement data alone. However, existing approaches fail
to obtain competitive performances in the problems of image super-resolution
and deblurring, which play a key role in most imaging systems. In this work, we
show that invariance to roto-translations is insufficient to learn from
measurements that only contain low-frequency information. Instead, we propose
scale-equivariant imaging, a new self-supervised approach that leverages the
fact that many image distributions are approximately scale-invariant, enabling
the recovery of high-frequency information lost in the measurement process.
We demonstrate throughout a series of experiments on real datasets that the
proposed method outperforms other self-supervised approaches, and obtains
performances on par with fully supervised learning.
\end{abstract}


% main content %

\bstctlcite{IEEEexample:BSTcontrol}

\section{Introduction}

Inverse problems are ubiquitous in scientific imaging and medical imaging. Linear inverse problems are generally modelled by the forward process
\begin{equation}
	y = A(z) + \varepsilon,
\end{equation}
where $y$ represents the observed measurements, $z$ represents the latent image,  $A$ is the imaging operator, and $\varepsilon$ depicts the noise corrupting the measurements. In this paper, we focus on additive white Gaussian noise but our results can easily be generalized for other noise distributions, e.g., Poisson, or mixed Gaussian-Poisson. Recovering $z$ from the noisy measurements is generally ill-posed, as the forward operator is often rank-deficient, e.g., when the number of measurements is smaller than the number of image pixels. Moreover, even when the forward operator is not rank-deficient, it might still be ill-conditioned, and reconstructions might be very sensitive to noise (for example, in image deblurring problems).

There is a wide range of approaches for solving the inverse problem, i.e., recovering $z$ from $y$. Classical methods rely on hand-crafted priors, e.g., a total variation, together with optimization algorithms to obtain a reconstruction that agrees with the observed measurements and the chosen prior~\cite{chambolle2016introduction},~\cite{dabov2007image}. These model-based approaches often obtain suboptimal results due to a mismatch between the true image distribution and the hand-crafted priors. An alternative approach leverages a dataset of reference images and associated measurements in order to train a reconstruction model~\cite{ongie2020deep}. While this approach generally obtains state-of-the-art performance, it is often hard to apply it to problems where ground truth images are expensive or even impossible to obtain.

In recent years, self-supervised methods have highlighted the possibility of learning from measurement data alone, thus bypassing the need for ground truth references. These methods focus either on taking care of noisy data and/or a rank-deficient operator. Noise can be handled through the use of Stein's unbiased risk estimator if the noise distribution is known~\cite{metzler_unsupervised_2020}, or via Noise2x techniques~\cite{krull_noise2void_2019,lehtinen_noise2noise_2018,batson_noise2self_2019} which require milder assumptions on the noise distribution such as independence across pixels~\cite{batson_noise2self_2019}. Rank-deficient operators can be handled using measurement data associated with multiple measurement operators $A_1,\dots,A_G$ (each with a different nullspace)~\cite{bora2018ambientgan},~\cite{tachella_sensing_2023}, or by assuming that the image distribution is invariant to certain groups of transformations such as translations or rotations~\cite{chen_equivariant_2021,chen_robust_2022}. This assumption is generally valid as a shifted or rotated image is generally still a valid image.

Self-supervised image deblurring methods have been proposed in settings where the images are blurred with multiple kernels~\cite{quan_learning_2022}, however this assumption is not verified in many applications. Moreover, self-supervised methods that rely on a single operator (e.g., a single blur kernel)~\cite{chen_equivariant_2021,chen_robust_2022}, have not been applied to inverse problems with a single forward operator that removes high-frequency information such as image super-resolution and image deblurring, except for methods exploiting internal/external patch recurrence~\cite{shocher_zero-shot_2018} which assume noiseless measurements, and thus do not easily generalize to other inverse problems.

Image deblurring and super-resolution are two inverse problems that are particularly difficult to solve using self-supervised methods. The former is severely ill-conditioned, and the latter is rank-deficient.
This makes methods that minimize the measurement consistency error inapplicable as they cannot retrieve any information that lies within the nullspace~\cite{chen_equivariant_2021} or where the forward operator is ill-conditioned, whenever there is even the slightest amount of noise present.

% Layout fix
\IEEEpubidadjcol

In this work, we propose a new self-supervised loss for training reconstruction networks for super-resolution and image deblurring problems which relies on the assumption that the underlying image distribution is scale-invariant, and can be easily applied to any linear inverse problem, as long as the forward process is known. \Cref{fig:Figure1} shows an overview of that loss. The main contributions are the following:

\begin{figure*}[!t]
	\includegraphics[width=\textwidth]{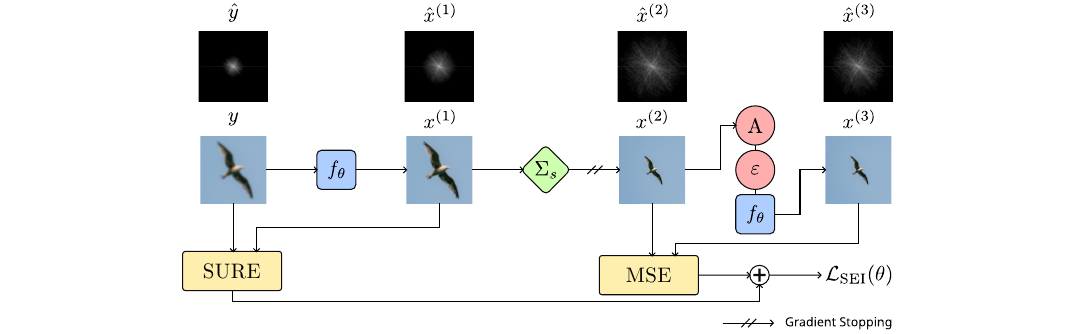}
	\caption{
	  \textbf{Overview of the method.}
	  Our scale-equivariant imaging loss $\mathcal{L}_{\mathrm{SEI}}$ combines a measurement consistency loss that is robust to noise (SURE) with a new self-supervised loss that substitutes downscaled reconstructions for ground truth images. The use of downscaling widens the frequency spectrum past the cut-off imposed by the blurring and super-resolution operators, thereby enabling the recovery of high-frequency information lost in the measurement process.
	}
	\label{fig:Figure1}
\end{figure*}

\begin{enumerate}
	\item We introduce a new theoretical framework proving that scaling transforms play a key role in learning from bandlimited images which includes the important problems of deblurring and super-resolution.

	\item We propose a new self-supervised loss for image deblurring and super-resolution leveraging the scale-invariance of typical image distributions to recover high-frequency information directly from corrupted images.

	\item We conduct a series of experiments on images from multiple image distributions corrupted using varying blur filters and downsampling operators, showing that scale-equivariant imaging (SEI) performs on par with fully-supervised models and outperforms other state-of-the-art self-supervised methods.
\end{enumerate}

\section{Related work}

\paragraph{Supervised approaches} Most state-of-the-art super-resolution and image deblurring methods rely on a dataset containing ground truth references \cite{xu_deep_2014, ren_deep_2018, dong_deep_2020, liang_swinir_2021}.  However, in many scientific imaging problems such as fluorescence microscopy and astronomical imaging, having access to high-resolution images can be difficult, or even impossible~\cite{belthangady2019applications}.

\paragraph{Prior-based approaches} A family of super-resolution and image deblurring methods rely on plug-and-play or diffusion priors~\cite{kamilov2023plug,zhao22Discrete,zhao23DDFM} and optimization methods. Other approaches rely on the inductive bias of certain convolutional architectures, i.e., the deep image prior \cite{ulyanov_deep_2018,ren2020neural}. Some methods combine both techniques \cite{mataev_deepred_2019}. However, these approaches can provide sub-optimal results if the prior is insufficiently adapted to the problem.

\paragraph{Stein's Unbiased Risk Estimate}
Stein's unbiased risk estimate is a popular strategy for learning from noisy measurement data alone, which requires knowledge of the noise distribution and of the forward model~\cite{stein_estimation_1981,eldar_generalized_2009,metzler_unsupervised_2020}.
However, for image super-resolution and image deblurring, the forward operator is typically either rank-deficient or at least severely ill-conditioned, and the SURE estimates can only capture discrepancies in the range of the forward operator.

\paragraph{Noise2x}
The Noise2Noise method~\cite{lehtinen_noise2noise_2018} shows that it is possible to learn a denoiser in a self-supervised fashion if two independent noisy realizations of the same image are available for training, without explicit knowledge of the noise distribution. Related methods extend this idea to denoising problems where a single noise realization is available, e.g., Noise2Void~\cite{krull_noise2void_2019} and Noise2Self~\cite{batson_noise2self_2019}, and to general linear inverse problems, e.g., Noise2Inverse~\cite{hendriksen_noise2inverse_2020}. However, all of these methods crucially require that the forward operator is invertible and well-conditioned, which is not the case for super-resolution and most deblurring problems.

\paragraph{Equivariant imaging}

Typical image distributions $\mathcal Z$ are invariant to rotations and/or translations, i.e., $\mathop{T_g} z \in \mathcal Z,\ \forall z \in \mathcal Z$, for every roto-translation $\{ T_g \}_{g = 1}^{|G|}$. Equivariant imaging (EI) leverages this property to learn reconstruction functions from measurement data alone, in problems where the forward operator is incomplete such as sparse-angle tomography, image inpainting, image fusion and magnetic resonance imaging~\cite{chen_equivariant_2021,chen_robust_2022,fatania_nonlinear_2023,zhao24Equivariant}.
However, this approach requires the forward operator not to be equivariant to rotations or translations, i.e., $A(\mathop{T_g} z) \neq \mathop{T_g} A(z)$ for at least some $z$ and $T_g$.
The super-resolution and deblurring problems generally do not verify this condition, and existing EI methods cannot be used. The identification of invariant image distributions has been studied in a subsequent work~\cite{tachella_sensing_2023}, but only for discrete images under the action of a compact group, with guarantees for almost every generic operator. We introduce a different analysis imposed by the non-compact scaling group which operates on continuous images. This new framework leverages Fourier analysis instead of representation theory and yields guarantees of a stronger kind, covering every bandlimiting operator as opposed to almost every.

\paragraph{Patch recurrence methods} Internal patch recurrence of natural images can be exploited to learn to super-resolve images without high-resolution references~\cite{shocher_zero-shot_2018}. However, while this self-supervised learning method can be extended to collections of low-resolution images via external patch recurrence,
it cannot be easily generalized to other inverse problems where the degraded image has different statistics from the ground truth image (e.g., image deblurring and noisy measurement data).
The method proposed in this work generalizes the idea of patch recurrence to scale-invariance and can be applied to any imaging inverse problem where high-frequency information is missing.

\section{Learning from bandlimited images}
\label{section.theory}

Imaging inverse problems are solved by finding a reconstruction function $f : y \mapsto z$ that maps the observed measurements $y \in \mathcal Y$ to the corresponding latent images $z \in \mathcal Z$. Self-supervised algorithms attempt to learn this function from a dataset of measurements $\{y_1, \ldots, y_N\} \subset \mathcal Y$ which raises the question of whether it is theoretically possible in the first place.  The analysis is typically broken down in two subproblems~\cite{tachella_sensing_2023}:
\begin{itemize}
	\item \textbf{Model identification.} Is the set of latent images $\mathcal Z$ fully determined by the set of measurements $\mathcal Y$?
	\item \textbf{Image recovery.} If the set of images $\mathcal Z$ is fully identified, is there a reconstruction function $f : y \mapsto z$ mapping measurements back to the corresponding latent image?
\end{itemize}
In this section, we focus on the problem of model identification for bandlimiting operators (\eg{} blur filters and downsampling operators) and show that scale-invariant signal sets $\mathcal Z$ are identifiable up to an arbitrary (fixed) frequency, unlike roto-translation invariant ones.
Unique image recovery given a known image set $\mathcal Z$ is possible, under appropriate conditions, if the set is low-dimensional. This has been studied in the compressive sensing literature (albeit generally under the assumption of random forward operators) \cite{bourrier14Fundamental} but is also tacitly assumed in any image reconstruction process. We will similarly make this assumption here.

We consider images as functions $z : \mathbb R^2 \to \mathbb C^d$ where $d \geq 1$ denotes the number of channels. Formally, we model the set of images as a subset $\mathcal Z$ of the Schwartz space\footnote{The Schwartz space $\mathcal S$ is dense in $L^2$, and unlike $L^2$ it is closed under convolution, making it a very general model for studying bandlimited inverse problems.} $\mathcal S$ of functions mapping $\mathbb R^2$ to $\mathbb C^d$. The problems of super-resolution and deblurring fall under the umbrella of image recovery from bandlimited measurements. Thus in this section, to simplify notation we assume that the forward operator is modelled as $y = \varkernel * z$, where $\varkernel \in \mathcal S$ is a filter acting on images $z$ by convolution.
 Moreover, we assume that $\varkernel$ has a finite bandwidth $0 \leq \xi_\varkernel < \infty$, defined as:
\begin{definition}[Bandlimited, bandwidth]
	Let $\varphi \in \mathcal S$, we say that $\varphi$ is bandlimited if for some $\xi_\varphi \geq 0$,
	\begin{equation}
		\hat \varphi(\xi) = 0, \forall \xi \in \mathbb R^2, \lVert \xi \rVert_2 \geq \xi_\varphi,
	\end{equation}
	where
	$
		\hat \varphi(\xi) = \int_{\mathbb R^2} \varphi(\varpos) e^{-i2\pi \xi^\top \varpos} \, \mathrm d\varpos
		$
	is the Fourier transform of $\varphi$. In that case, we call the smallest such $\xi_\varphi$ the bandwidth of $\varphi$.
\end{definition}

While physical images $z$ are typically not bandlimited, the measurements $y=\varkernel * z$ have bandwidth at most $\xi_{\varkernel}$, and due the Shannon-Nyquist's sampling theorem, they can be represented by discrete samples on a regular grid of sampling rate $2 \xi_{\varkernel}$~\cite{vetterli2014foundations}.

In the self-supervised learning setting, we only observe the set of bandlimited measurements
\begin{equation} \label{eq:def:Y}
	\mathcal Y = \{ \varkernel * z,\ z \in \csignalset \}.
\end{equation}
and the key question is whether we can recover the set of images $\csignalset$ from the set of measurements $\mathcal Y$. This formulation has been studied when $\csignalset$ is a subset of the finite-dimensional space $\mathbb R^n$ under the assumption that ground truth images have finite bandwidth~\cite{tachella_sensing_2023}. Here, we consider the more general setting where $\csignalset$ contains continuous images that can have an infinite bandwidth. According to Shannon-Nyquist's sampling theorem~\cite{vetterli2014foundations}, it is impossible to represent them faithfully on a (finite) grid of pixels, no matter how fine it is. This makes the formulation of the problem impractical to work with and instead we propose to recover the set
\begin{equation} \label{eq:def:X}
	\mathcal X = \{ \varphi * z,\ z \in \mathcal Z \},
\end{equation}
which contains low-pass filtered versions $x=\varphi * z$ of the images in $\csignalset$, using a bandlimiting filter $\varphi \in \mathcal S$ with the desired increased bandwidth $\xi_\varphi > \xi_\varkernel$. In this way, the images in $\dsignalset$ can be represented in a discrete grid with sampling rate $2 \xi_\varphi$. In the context of image deblurring, $\xi_\varkernel$ represents the resolution of the blurry images and $\xi_\varphi$ the target resolution, and for image super-resolution the ratio $\xi_\varphi / \xi_\varkernel$ can be interpreted as the upsampling rate.

In general, the set $\dsignalset$ cannot be entirely known from the observation set $\mathcal Y$ as it only contains information about the lower-frequency content of the latent images $z \in \csignalset$, and none for the frequencies in $\xi_{\varkernel} < \xi \leq \xi_\varphi$. However, if $\csignalset$ is known to be invariant to a group action $T_g : \mathcal S \mapsto \mathcal S$, with $g \in G$ then the set of observed images is equal to the augmented set of observed images \cite{tachella_sensing_2023}
\begin{equation}
	\mathcal Y_{G} = \{ \varkernel * T_g z,\ z \in \csignalset, \, g \in G \}.
\end{equation}
In this case, depending on the specific group action, $\mathcal Y_{\mathcal G}$ might have the entirety of the information of $\csignalset$ up to the higher frequency $\xi_\varphi > \xi_\varkernel$. This amounts to saying that the set $\dsignalset$ can be recovered from the set $\mathcal Y_{G}$. The three most common invariances to group actions of natural image distributions are rotations, translations and scaling transforms which are defined as follows:

\begin{definition}[Scaling transforms, translations and rotations]
	Let $z \in \mathcal S$. We define the action by the groups of scaling transforms $\Sigma$, translations $\tau$ and rotations $R$ by
	\begin{align}
		(\mathop{\Sigma_s} z)(\varpos) &= z(s^{-1}\varpos),\ \varpos \in \mathbb R^2, s \in (0, \infty),  \\
		(\mathop{\tau_\varposalt} z)(\varpos) &= z(\varpos - \varposalt),\ \varpos \in \mathbb R^2, \varposalt \in \mathbb R^2, \\
		(\mathop{R_\theta} z)(\varpos) &= z(P_{\theta}^{-1} \ \varpos), \varpos \in \mathbb R^2, \theta \in [0, 2\pi),
	\end{align}
	where $P_{\theta} \in \mathbb R^{2 \times 2}$ is the rotation matrix for angle $\theta$.
\end{definition}

\begin{figure}[!t]
	\includegraphics[width=\columnwidth]{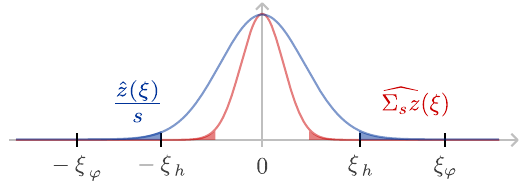}
	\caption{\textbf{Spectral effect of rescaling.} In bandlimited problems such as image super-resolution and deblurring, the image frequencies are only observed on a low-frequency band $[0, \xi_h)$ determined by the bandlimiting operator (\eg{} blur and anti-aliasing filters), and the goal is to recover texture information in a higher frequency band $(\xi_h, \xi_\varphi)$ where $\xi_\varphi / \xi_h$ can be interpreted as the increase in effective resolution. Scale-invariance makes this possible as the spectral information is available at all scales: an image $z(\varpos)$ and the rescaled image $\Sigma_s z(\varpos)$ with scaling factor $s$ have the same spectral content except in different frequency bands.}
	\label{fig:Scaling_Transforms}
\end{figure}

The choice of the group $G$ plays a fundamental role in the recovery of the information lost in the measurement process. While most previous works have focused on rotations and translations~\cite{chen_equivariant_2021,chen_robust_2022,tachella_sensing_2023}, here we show that these are not sufficient for learning from bandlimited measurements alone:

\begin{restatable}{theorem}{theoremSpecialEuclidean} \label{theorem:special_euclidean}
	Let $\varkernel, \varphi \in \mathcal S$. If $\varkernel$ and $\varphi$ are bandlimited with bandwidths $\xi_\varkernel < \xi_\varphi$, then there are sets $\mathcal Z_1, \mathcal Z_2 \subseteq \mathcal S$ such that
	\begin{itemize}
		\item $\mathcal Z_1$ and $\mathcal Z_2$ are roto-translation invariant, that is
			\begin{equation}
				\tau_\varposalt R_\theta z \in \mathcal Z_\ell,\ z \in \mathcal Z_\ell, \varposalt \in \mathbb R^2, \theta \in [0, 2\pi), \ell \in \{ 1, 2 \},
			\end{equation}
		\item for $\mathcal Y_1 = h * \mathcal Z_1$ and $\mathcal Y_2 = h * \mathcal Z_2$ as in \cref{eq:def:Y},
			\begin{equation}
				\mathcal Y_1 = \mathcal Y_2
			\end{equation}
		\item for $\mathcal X_1 = \varphi * \mathcal Z_1$ and $\mathcal X_2 = \varphi * \mathcal Z_2$ as in \cref{eq:def:X},
			\begin{equation}
				\mathcal X_1 \neq \mathcal X_2.
			\end{equation}
	\end{itemize}
\end{restatable}

For the sake of brevity, we only exhibit one example of such sets $\mathcal Z_1$ and $\mathcal Z_2$ in the proof (in the appendix) but it is possible to find a matching $\mathcal Z_2$ for any candidate $\mathcal Z_1$ by altering it in the frequency band $(\xi_\varkernel, \xi_\varphi)$.

Intuitively the action by rotations and translations does not increase the bandwidth of an image and thus does not provide information in the frequency range $\xi_\varkernel < \xi < \xi_\varphi$. On the contrary, the action by scaling modifies the bandwidth of an image (illustrated in \Cref{fig:Figure1,fig:Scaling_Transforms}), which allows us to identify the set $\dsignalset$ from the bandlimited images alone.

\begin{restatable}{theorem}{theoremHomothety} \label{theorem:homothety}
	Let $\varkernel, \varphi \in \mathcal S$ and $\mathcal Z \subseteq \mathcal S$. If $\mathcal Z$ is scale-invariant, that is
	\begin{equation}
		\mathop{\Sigma_s} z \in \mathcal Z, \forall z \in \mathcal Z, s \in (0, \infty),
	\end{equation}
	$\varphi$ is bandlimited, and $\hat \varkernel(0) \neq 0$, where $\hat \varkernel$ is the Fourier transform of $\varkernel$, then the set $\mathcal X$ defined in \cref{eq:def:X} is uniquely determined by the set of observed images $\mathcal Y$ defined in \cref{eq:def:Y}. More precisely, there is a $s \geq \frac{\xi_\varphi}{\xi_h}$ which only depends on $\varphi$ and $\varkernel$, for which $\hat h(s^{-1} \xi) \neq 0$ as long as $\lVert \xi \rVert_2 < \xi_\varphi$, and
	\begin{equation}
		\mathcal X = \left\{ x_y, \ y \in \mathcal Y \right\},
	\end{equation}
	where $x_y \in \mathcal S$ is defined by
	\begin{equation} \label{eq:def:xy}
		\hat x_y(\xi) = \begin{cases}
			\frac{\hat \varphi(\xi) \hat y(s^{-1} \xi)}{s\hat h(s^{-1}\xi)}, & \text{if $\lVert \xi \rVert_2 < \xi_\varphi$}, \\
			0, & \text{otherwise}.
		\end{cases}
	\end{equation}
\end{restatable}

See the appendix for the detailed proofs. These results show that leveraging scale-invariance is crucial for learning from bandlimited measurements alone, and in the next section we present a practical self-supervised loss that exploits this property.

\section{Proposed method}
\label{section.method}

In this section, we assume that there is an approximately one-to-one correspondence between the latent images and the measured images which guarantees that image recovery is possible, and we introduce our proposed method for deblurring and super-resolution that exploits the scale-invariance of typical image distributions.

As shown by \Cref{theorem:homothety}, as long as the set of latent images is invariant to scale, there is enough information in the set of bandlimited images to uniquely determine the set of images with an arbitrarily high bandwidth. This hints that self-supervised learning can be as effective as supervised learning for solving bandlimiting imaging problems like image super-resolution and deblurring.

In the problems of deblurring and super-resolution, the goal is to recover a high-resolution image $x \in \mathbb R^n$ from a low-resolution image $y \in \mathbb R^m$. For notational simplicity, we represent images as vectors where every dimension corresponds to a different pixel but the method is designed to be used to learn a nonlinear reconstruction filter that takes in images of arbitrary size, \eg{} a convolutional neural network (CNN) or a suitable vision transformer (ViT). The relation between $x$ and $y$ is modelled as
\begin{equation}
	y = (\varkernel * x) {\downarrow_r} + \varepsilon
\end{equation}
where $\downarrow_r$ is the subsampling operator by factor $r$, and $\varepsilon \sim \mathcal N(0, \sigma^2 I)$ is additive Gaussian noise\footnote{We focus on Gaussian noise in this work and use the corresponding risk estimator -- it can be substituted for other noise distributions.}. In super-resolution problems $\varkernel$ is the anti-aliasing filter and $r > 1$ whereas in deblurring problems $\varkernel$ is the blur kernel and $r = 1$. Self-supervised learning-based methods attempt to solve this problem by learning a reconstruction function\footnote{For the problem to be well-defined, it is necessary that there is one-to-one-ness between the large-bandwidth images and the small-bandwidth images.} $f_\theta : y \mapsto x$ parametrized by a neural network with learnable weights $\theta \in \mathbb R^p$ using a dataset of blurry/low-resolution images only\,$\{y_i\}_{i=1}^{N}$.

\paragraph{Stein's unbiased risk estimator}

The SURE loss~\cite{metzler_unsupervised_2020,chen_robust_2022} is a self-supervised loss that serves as an unbiased estimator of the mean squared error (MSE) in measurement space
\begin{equation}
	\mathbb E_{x,y} \lVert \left(\varkernel * f_\theta(y)\right){\downarrow_r} - (\varkernel * x){\downarrow_r} \rVert_2^2 ,
\end{equation}
which does not require access to ground truth images $x$. The SURE loss for measurements corrupted by Gaussian noise of standard deviation $\sigma$ \cite{chen_robust_2022} is given by:
\begin{equation} \label{eq.sure}
	\begin{split}
		\mathcal L_{\text{SURE}} (\theta) = & \sum_{i=1}^N  \frac 1 m \lVert \left(\varkernel * f_\theta(y_i)\right){\downarrow_r}  - y_i \rVert_2^2 \\ &- \sigma^2
  + \frac{2\sigma^2}{m} \text{div}[\left(\varkernel * f_\theta\right){\downarrow_r}] (y_i).
	\end{split}
\end{equation}
where the divergence $\text{div}$ is approximated using a Monte-Carlo estimate~\cite{ramani2008monte}.
It is worth noting that, in the noiseless setting, i.e., $\sigma = 0$, the SURE loss is equal to a measurement consistency constraint $\left(\varkernel * f_\theta(y)\right){\downarrow_r}=y$,
\begin{equation}
	\mathcal L_{\text{MC}} (\theta) =  \sum_{i=1}^N \frac 1 m \lVert \left(\varkernel * f_\theta(y_i)\right){\downarrow_r} - y_i \rVert_2^2.
\end{equation}
which does not penalize the reconstructions in the nullspace of $A$. In the setting of super-resolution and image deblurring, the nullspace of the operator is associated with the high-frequencies of the image $x$, and thus training with the SURE loss alone fails to reconstruct the high-frequency content.

We propose to minimize the training loss
\begin{equation} \label{eq.sei}
	\mathcal L_{\text{SEI}}(\theta) = \mathcal L_{\text{SURE}} (\theta) + \alpha \mathcal L_{\text{SEQ}} (\theta),
\end{equation}
where $\alpha > 0$ is a trade-off parameter, $\mathcal L_{\text{SURE}} (\theta)$ is given in \Cref{eq.sure}, and $\mathcal L_{\text{SEQ}}(\theta)$ is the equivariant loss introduced below, which is inspired by the robust equivariant imaging loss \cite{chen_robust_2022}. In all experiments, we set $\alpha = 1$.
The two contributions of our method are the use of the action by downscaling transforms instead of the action by shifts or by rotations, and stopping the gradient during the forward pass, as illustrated in~\Cref{fig:Figure1} (details are explained below).

In the forward pass, we first compute $x_{i,\theta}^{(1)} = f_\theta(y_i)$, which we view as an estimation of the $i$th ground truth image $x_i$. Next, we compute $x_i^{(2)} = \mathop{\Sigma_s} x_{i,\theta}^{(1)}$, where $\Sigma_s$ is the downscaling transform by factor $s \in (0, 1)$, without propagating the gradient, which effectively leaves $\nabla_\theta x_i^{(2)} = 0$. Finally, we compute $x_{i,\theta}^{(3)} = f_\theta(Ax_i^{(2)} + \tilde \varepsilon)$ where $\tilde \varepsilon$ is sampled from $\mathcal N(0, \sigma^2I)$. The three terms are used to compute $\mathcal L_{\text{SURE}}$ and $\mathcal L_{\text{SEQ}}$, where $\mathcal L_{\text{SEQ}}$ is defined as
\begin{equation}
	\mathcal L_{\text{SEQ}} (\theta) = \sum_{i=1}^N \frac 1 n \lVert x_{i,\theta}^{(3)} - x_i^{(2)} \rVert_2^2.
\end{equation}

\paragraph{Downscaling transformations} In all experiments, $\Sigma_s$ consists in resampling the input image viewed as periodic on a (coarse) grid with a spacing of $s^{-1}$ pixels, using bicubic interpolation. The grid center is chosen randomly at each call and the scale factor is chosen uniformly from $\{ .5, .75 \}$. \Cref{fig:Downscaling_Transforms} shows two possible downscalings of an image patch.

\paragraph{Gradient stopping}
The gradient-stopping step, which removes the dependence of $x_i^{(2)}$ on the network weights $\theta$ plays an important role in the self-supervised loss. \Cref{tab:Ablation_Study} shows that gradient stopping significantly improves the performance of the self-supervised network for both image deblurring and image super-resolution, i.e., by about 1dB. Such gradient-stopping steps have proved crucial in various self-supervised representation learning methods and lead to significant performance increase \cite{grill_bootstrap_nodate,chen_exploring_2020}.

\begin{table*}
	\caption{Deblurring and super-resolution performance. (Bold) best performing method, (blue) second best.
	}
	\label{tab:Main}

	\centering
	\setlength{\tabcolsep}{2pt} % default: 6pt
	\begin{tabular}{lcccccccccccccccccc}
		\toprule
& \multicolumn{3}{c}{Gaussian filter (small)} & \multicolumn{3}{c}{Gaussian filter (medium)} & \multicolumn{3}{c}{Gaussian filter (large)} & \multicolumn{3}{c}{Box filter (small)} & \multicolumn{3}{c}{Box filter (medium)} & \multicolumn{3}{c}{Box filter (large)} \\
	\cmidrule(rl){2-4} \cmidrule(rl){5-7} \cmidrule(rl){8-10} \cmidrule(rl){11-13} \cmidrule(rl){14-16} \cmidrule(rl){17-19}
		Method & PSNR & SSIM & LPIPS & PSNR & SSIM & LPIPS & PSNR & SSIM & LPIPS & PSNR & SSIM & LPIPS & PSNR & SSIM & LPIPS & PSNR & SSIM & LPIPS \\
		\midrule
		Supervised & \pmb{$30.7$} & \pmb{$0.923$} & \textcolor{blue}{$0.064$} & \pmb{$25.9$} & \pmb{$0.803$} & \pmb{$0.249$} & \pmb{$23.7$} & \pmb{$0.694$} & \pmb{$0.361$} & \pmb{$29.1$} & \pmb{$0.879$} & \pmb{$0.107$} & \pmb{$27.3$} & \pmb{$0.825$} & \pmb{$0.171$} & \pmb{$25.4$} & \pmb{$0.755$} & \pmb{$0.255$} \\
		SEI (Ours) & \textcolor{blue}{$30.3$} & \textcolor{blue}{$0.917$} & \pmb{$0.062$} & \textcolor{blue}{$25.9$} & \textcolor{blue}{$0.795$} & \textcolor{blue}{$0.254$} & \textcolor{blue}{$23.7$} & \textcolor{blue}{$0.687$} & \textcolor{blue}{$0.379$} & \textcolor{blue}{$28.8$} & \textcolor{blue}{$0.874$} & \textcolor{blue}{$0.112$} & \textcolor{blue}{$27.0$} & \textcolor{blue}{$0.818$} & \textcolor{blue}{$0.178$} & \textcolor{blue}{$25.0$} & \textcolor{blue}{$0.739$} & \textcolor{blue}{$0.309$} \\
		CSS & $29.1$ & $0.899$ & $0.202$ & $24.2$ & $0.707$ & $0.462$ & $22.8$ & $0.622$ & $0.549$ & $26.0$ & $0.804$ & $0.282$ & $24.2$ & $0.713$ & $0.418$ & $23.1$ & $0.640$ & $0.508$ \\
		SURE & $26.3$ & $0.731$ & $0.257$ & $23.2$ & $0.594$ & $0.447$ & $21.1$ & $0.490$ & $0.538$ & $25.8$ & $0.754$ & $0.234$ & $24.5$ & $0.675$ & $0.329$ & $23.3$ & $0.595$ & $0.437$ \\
		BM3D & $29.5$ & $0.881$ & $0.184$ & $25.4$ & $0.757$ & $0.380$ & $23.5$ & $0.649$ & $0.506$ & $27.1$ & $0.809$ & $0.270$ & $25.5$ & $0.743$ & $0.356$ & $24.6$ & $0.700$ & $0.420$ \\
		DIP & $27.3$ & $0.825$ & $0.177$ & $24.4$ & $0.691$ & $0.350$ & $22.8$ & $0.579$ & $0.452$ & $24.1$ & $0.675$ & $0.378$ & $23.1$ & $0.598$ & $0.433$ & $22.3$ & $0.538$ & $0.481$ \\
		EI & $28.0$ & $0.835$ & $0.247$ & $24.8$ & $0.718$ & $0.360$ & $22.7$ & $0.594$ & $0.445$ & $24.5$ & $0.709$ & $0.344$ & $23.2$ & $0.646$ & $0.385$ & $22.9$ & $0.611$ & $0.463$ \\
\midrule
		Blurry images & $26.4$ & $0.794$ & $0.269$ & $22.8$ & $0.597$ & $0.461$ & $21.2$ & $0.494$ & $0.604$ & $23.4$ & $0.629$ & $0.359$ & $21.9$ & $0.528$ & $0.464$ & $21.0$ & $0.465$ & $0.552$ \\
\bottomrule
	\end{tabular}

	\vspace{.2cm}

	\setlength{\tabcolsep}{11.3pt} % default: 6pt
	\begin{tabular}{lccccccccc}
	\toprule
	& \multicolumn{3}{c}{Downsampling ($\times 2$)} & \multicolumn{3}{c}{Downsampling ($\times 3$)} & \multicolumn{3}{c}{Downsampling ($\times 4$)} \\
	\cmidrule(rl){2-4} \cmidrule(rl){5-7} \cmidrule(rl){8-10}
	Method & PSNR $\uparrow$ & SSIM $\uparrow$ & LPIPS $\downarrow$ & PSNR $\uparrow$ & SSIM $\uparrow$ & LPIPS $\downarrow$ & PSNR $\uparrow$ & SSIM $\uparrow$ & LPIPS $\downarrow$ \\
	\midrule
	Supervised & \pmb{$28.99$} & \pmb{$0.8821$} & \pmb{$0.1416$} & \pmb{$25.94$} & \pmb{$0.7833$} & \pmb{$0.2491$} & \pmb{$24.25$} & \pmb{$0.6983$} & \pmb{$0.3185$} \\
	SEI (Ours) & $28.09$ & $0.8500$ & $0.1972$ & \textcolor{blue}{$25.29$} & $0.7436$ & $0.3833$ & \textcolor{blue}{$23.73$} & \textcolor{blue}{$0.6568$} & $0.4642$ \\
	CSS & \textcolor{blue}{$28.50$} & \textcolor{blue}{$0.8616$} & \textcolor{blue}{$0.1781$} & $25.29$ & \textcolor{blue}{$0.7441$} & $0.3835$ & $21.38$ & $0.6126$ & $0.5085$ \\
	DIP & $25.77$ & $0.7673$ & $0.2678$ & $24.11$ & $0.6828$ & \textcolor{blue}{$0.3649$} & $22.79$ & $0.5915$ & \textcolor{blue}{$0.4447$} \\
	EI & $27.76$ & $0.8413$ & $0.2469$ & $25.13$ & $0.7387$ & $0.4108$ & $23.51$ & $0.6441$ & $0.5132$ \\
	\midrule
	Bicubic & $27.23$ & $0.8306$ & $0.2986$ & $24.67$ & $0.7185$ & $0.4660$ & $23.26$ & $0.6325$ & $0.5882$ \\
	\bottomrule
	\end{tabular}
\end{table*}

\begin{figure}[!t]
	\includegraphics[width=\columnwidth]{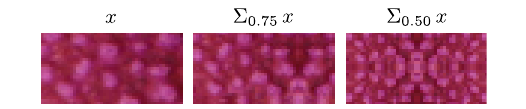}
	\caption{\textbf{Implementation of the (down-)scaling transformations.} Scale-equivariant imaging (SEI) uses rescaled versions of network outputs as learning targets implemented as bicubic resampling on a coarser grid randomly sampled for each batch element. Using a coarser grid makes the resulting image close to the target manifold $\mathcal X$ as its spectral information corresponds to the low-frequency content of the (unknown) target image thanks to the loss $\mathcal L_{\text{SURE}}$ and to the scale-invariance of the latent image set $\mathcal Z$.}
	\label{fig:Downscaling_Transforms}
\end{figure}

\paragraph{Comparison to self-supervision via patch recurrence}
A related self-supervised approach to ours is the zero-shot super-resolution proposed by~\cite{shocher_zero-shot_2018}, which uses the measurements $y_i$ as ground truth targets, and computes corresponding corrupted inputs applying the forward operator to $y$, i.e.,
\begin{equation}
    \tilde{y} = \left(\varkernel * y\right){\downarrow_r}  + \tilde \varepsilon
\end{equation}
where $\tilde \varepsilon$ is sampled from a zero-mean Gaussian distribution with the same standard deviation. Then, a self-supervised loss is constructed as
\begin{equation} \label{eq: CSS}
	\mathcal L_{\text{CSS}}(\theta) =  \sum_{i=1}^N \frac 1 m  \lVert y_i - f_\theta(\tilde{y}_i) \rVert^2_2
\end{equation}
which we name Classic Self-Supervised (CSS) loss.
 This approach is only valid when measurements $y$ can be interpreted as images, such as in a super-resolution problem \cite{shocher_zero-shot_2018}, however, it is unclear whether it is effective for solving general problems, e.g., compressed sensing.

The CSS method requires that the forward operator is a downsampling transformation and that the noise $\varepsilon$ is sufficiently small, such that
$\left(\varkernel * y\right){\downarrow_r} \approx x'$ for some $x'$ which also belongs to the dataset. On the contrary, as the proposed method exploits scale-invariance on the reconstructed images, it can be used with any operator $A$. Thus, for the super-resolution problems with a small amount of noise, we expect to obtain a similar performance with both methods, whereas in the more challenging image deblurring problem we expect to observe a different performance across methods.

\begin{table}
	\caption{Ablation study. GS: Gradient stopping.}
	\label{tab:Ablation_Study}

	\centering
	\setlength{\tabcolsep}{2.5pt} % default: 6pt
\begin{tabular}{cccccccccc}
\toprule
& \multicolumn{3}{c}{Gaussian filter (medium)} & \multicolumn{3}{c}{Box filter (medium)} & \multicolumn{3}{c}{Downsampling ($\times 2$)} \\
\cmidrule(rl){2-4} \cmidrule(rl){5-7} \cmidrule(rl){8-10}
GS & PSNR & SSIM & LPIPS & PSNR & SSIM & LPIPS & PSNR & SSIM & LPIPS \\
\midrule

\checkmark & \pmb{$25.9$} & \pmb{$0.795$} & \pmb{$0.254$} & \pmb{$27.0$} & \pmb{$0.818$} & \pmb{$0.178$} & \pmb{$28.1$} & \pmb{$0.850$} & \pmb{$0.197$} \\
\texttimes & $22.9$ & $0.632$ & $0.506$ & $22.6$ & $0.601$ & $0.551$ & $26.8$ & $0.810$ & $0.305$ \\
\bottomrule
\end{tabular}
\end{table}

\section{Experiments}
\label{section.experiments}

We evaluate our method for non-blind deblurring and super-resolution on images
corrupted by 1) Gaussian filters, 2) box filters, and 3) bicubic downsampling
operators. This allows us to verify that our method performs well across
different imaging modalities. We use SwinIR~\cite{liang_swinir_2021} as the
architecture of our reconstruction network -- one of the best-performing
architectures for image super-resolution. In each
setting, we provide a comparison with a wide variety of state-of-the-art
methods, including the CSS method in~\Cref{eq:
CSS}~\cite{shocher_zero-shot_2018}, the diffusion method DiffPIR
\cite{zhu2023denoising}, (robust) equivariant imaging (EI)
\cite{chen_robust_2022}, GSURE
\cite{eldar_generalized_2009,metzler_unsupervised_2020}, the Plug-and-Play (PnP)
method DPIR~\cite{zhang2017learning}, BM3D \cite{dabov2007image} and Deep Image
Prior (DIP) \cite{ulyanov_deep_2018, darestani2021accelerated}. For a fair
comparison and to demonstrate the importance of the training loss in our
method, we also compare against models trained the same way as ours
except for the training loss. In addition to the self-supervised baselines, we
compare with supervised baselines and show that our scale-equivariant imaging
approach (SEI) compares favorably against them, even though it is fully
self-supervised and does not have direct access to the underlying distribution
of clean images. Our results show that our method performs better than other
state-of-the-art self-supervised methods and significantly improves the quality
of corrupted images.

\begin{figure*}[!t]
	\includegraphics[width=\textwidth]{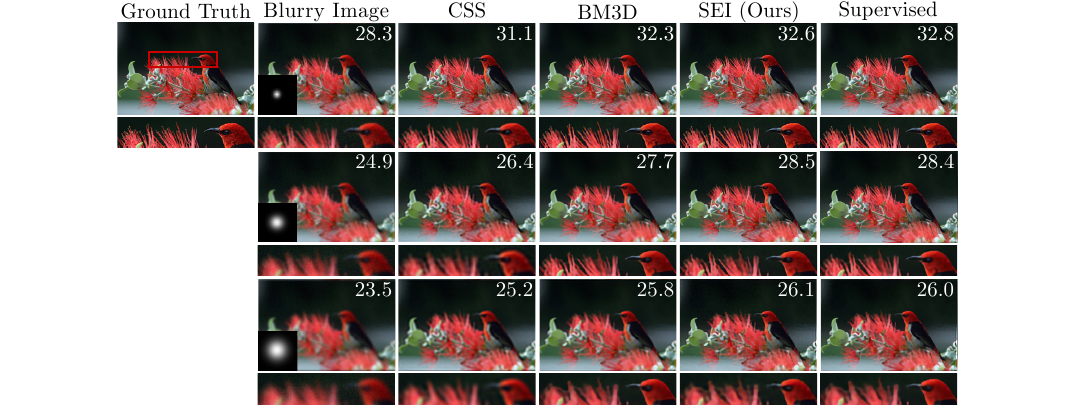}
	\caption{\textbf{Visual comparison for Gaussian deblurring.} Top-right corner: PSNR, bottom-left corner: point-spread functions.}
	\label{fig:Gaussian_Deblurring}
\end{figure*}

\begin{figure}[!t]
	\includegraphics[width=\columnwidth]{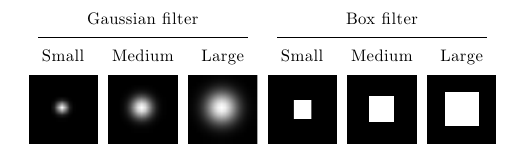}
	\caption{\textbf{Point-spread functions.} Our method makes no assumption about the shape or strength of the blur kernel making it broadly applicable and in particular we verify that it performs well on Gaussian filters with standard deviation $1$, $2$ or $3$ pixels, and for box filters with radius $2$, $3$ or $4$ pixels.}
	\label{fig:Blur_Kernels}
\end{figure}

For every degradation model, we constitute pairs of clean and corrupted images
that we use for training and evaluation. We evaluate all methods, using two
datasets of natural images DIV2K \cite{agustsson_ntire_2017} and Urban100
\cite{huang_single_2015}, and a dataset of computed tomography (CT) scans
\cite{armato2011lung}. \Cref{fig:Blur_Kernels} shows the 6
different blur kernels used for synthesizing some of the corrupted images, and
the remaining ones are synthesized using bicubic downsampling with multiple
sampling rates (2 and 4). The resulting images are further corrupted by additive
white Gaussian noise making the degradation model closer to real-world image
degradations happening in most imaging systems. The standard deviation for the
Gaussian noise is set to $\sigma = 5/255$. The 900 image pairs from DIV2K are
split into 100 testing pairs and 800 training pairs, the 100 image pairs from
Urban100 are split into 10 testing pairs and 90 training pairs, and the 5092
pairs for CT scans are split into 100 testing pairs and 4992 training pairs.
\Cref{tab:Main,tab:Additional_Datasets} show that
scale-invariant imaging tends to perform as well as the supervised baseline and
significantly better than the other state-of-the-art self-supervised methods,
both across different imaging modalities and different image distributions.
Examples of deblurred and upsampled images are shown in
\Cref{fig:Gaussian_Deblurring}.

\begin{figure}[!t]
	\includegraphics[width=\columnwidth]{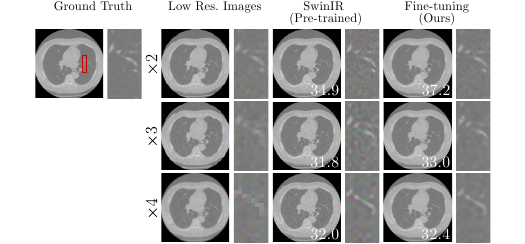}
	\caption{\textbf{Fine-tuning on CT scans.} Self-supervised
	methods for super-resolution do not require ground truth images and can
	be used to fine-tune pre-trained models on test data directly. Here, we
	fine-tune models trained on natural images to upsample CT scans corrupted by
	additional Gaussian noise.}
	\label{fig:Fine_Tuning}
\end{figure}

\paragraph{Fine-tuning}
Because it is fully self-supervised and unlike supervised methods, it is
possible to use scale-equivariant imaging directly on test data.
This enables fine-tuning models pre-trained on image distributions that might
differ from the test data, generally resulting in better performing models. We
evaluate this approach by fine-tuning 3 models on CT scans. The 3 models are
pre-trained on pairs high and low-resolution images from DIV2K (natural
images), each for a different sampling rate (2, 3, and 4), all with the same
supervised training algorithm. The low-resolution images are synthesized using
bicubic downsampling similarly to the low-resolution CT scans that we are using
except for the absence of additional noise. \Cref{tab:Fine_Tuning} shows that
the performance of the fine-tuned model is significantly greater than the
performance of the pre-trained model, ranging from a $.5$~dB to a $2$~dB
difference in PSNR. Example of upsampled images are shown in
\Cref{fig:Fine_Tuning}.

\paragraph{Ablation studies}
In addition to the experiments comparing our method with other methods, we also conduct a range of ablation studies to explore the role of the gradient stopping, the sensitivity of key parameters, and the influence of the chosen network architecture and interpolation strategy. Here, we will focus on the results for the gradient stopping and relegate the additional studies to~\Cref{section.additional_ablation_studies}.
For each of the three degradation models, we train two models using our proposed method: one with and one without gradient stopping. \Cref{tab:Ablation_Study} shows that gradient stopping improves the reconstructions by more than $2$~dB (PSNR) in all three cases. This is consistent with our broader empirical observation that gradient stopping consistently improves the performance of the proposed method.

\begin{figure*}[!t]
	\includegraphics[width=\textwidth]{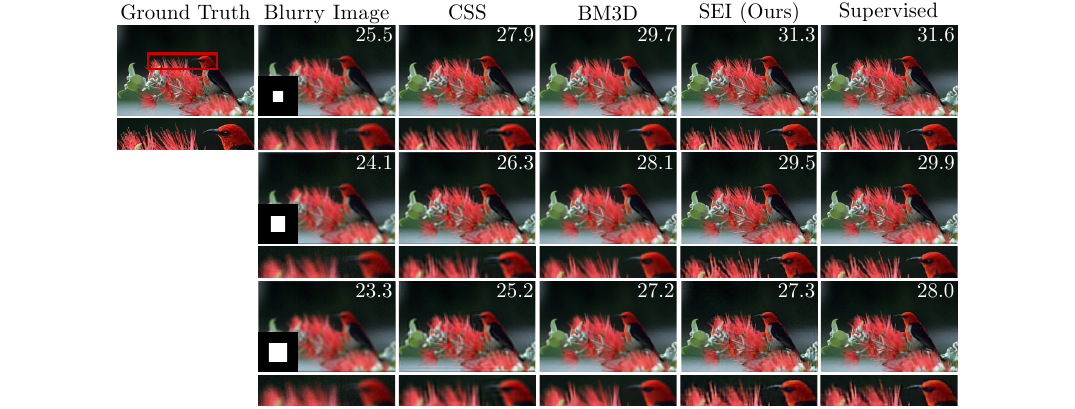}
	\caption{\textbf{Visual comparison for box deblurring.} Bottom-left corner: point-spread function, top-right corner: PSNR.}
	\label{fig:Box_Deblurring}
\end{figure*}

\begin{table}
	\caption{Fine-tuning performance. FT: Fine-tuning.}
	\label{tab:Fine_Tuning}

	\centering
	\setlength{\tabcolsep}{3pt} % default: 6pt
\begin{tabular}{lccccccccc}
\toprule
& \multicolumn{3}{c}{Downsampling ($\times 2$)} & \multicolumn{3}{c}{Downsampling ($\times 3$)} & \multicolumn{3}{c}{Downsampling ($\times 4$)} \\
\cmidrule(rl){2-4} \cmidrule(rl){5-7} \cmidrule(rl){8-10}
FT & PSNR & SSIM & LPIPS & PSNR & SSIM & LPIPS & PSNR & SSIM & LPIPS \\
\midrule
$\checkmark$ & \pmb{$37.1$} & \pmb{$0.882$} & \pmb{$0.123$} & \pmb{$33.4$} & \pmb{$0.870$} & \pmb{$0.138$} & \pmb{$32.9$} & \pmb{$0.835$} & \pmb{$0.194$} \\
\texttimes & $34.9$ & $0.804$ & $0.260$ & $32.0$ & $0.799$ & $0.370$ & $32.4$ & $0.816$ & $0.398$ \\
\bottomrule
\end{tabular}
\end{table}

\subsection{Experimental details}

\paragraph{Performance metrics}
Similarly to \cite{liang_swinir_2021}, we convert reference images and image estimates from RGB to YCbCr, and we report the PSNR and the SSIM for the luminance (Y) channel. For LPIPS, the values are directly obtained from the RGB images as usual.

\paragraph{Training hyperparameters}
We train all of the models using a batch size of $8$, for $500$ epochs for DIV2K, for $4000$ epochs for Urban100, and for $100$ epochs for the CT scans. For all trainings except fine-tuning ones, we use Adam~\cite{kingma_adam_2014} with an initial learning rate of $0.0005$ for deblurring, and of $0.0002$ for super-resolution, with $\beta = (0.9, 0.999)$ in both cases. For fine-tuning, we use regular stochastic gradient descent (SGD) instead, with an initial learning rate of $0.01$. This leads to more stable training dynamics and better resulting performance. Each model takes from 10 to 20 hours to train on a single recent GPU using SIDUS~\cite{quemener_sidus_2013}.

\paragraph{Baseline hyperparameters}
For the supervised baseline (Supervised), we use the $\ell^2$-distance as the training loss. For deep image prior (DIP), we use a learning rate of $0.005$ and random inputs of size $16 \times 16$ with $32$ channels. We use the network architecture proposed by \cite{darestani2021accelerated}, for at most $100$ iterations and with early-stopping. For plug-and-play (PnP), we use the method proposed by \cite{zhang2017learning}. The denoiser is pre-trained on images from various datasets \cite{ma2016waterloo,timofte2017ntire}, including DIV2K.

\begin{table}
\caption{Deblurring performance on additional datasets.}
\label{tab:Additional_Datasets}
\centering
\setlength{\tabcolsep}{3pt} % default: 6pt
\begin{tabular}{lcccccccc}
\toprule
& \multicolumn{3}{c}{Urban 100} & \multicolumn{3}{c}{CT scans} \\
 \cmidrule(rl){2-4} \cmidrule(rl){5-7}
Method & PSNR $\uparrow$ & SSIM $\uparrow$ & LPIPS $\downarrow$ & PSNR $\uparrow$ & SSIM $\uparrow$ & LPIPS $\downarrow$ \\
\midrule
Supervised & \textcolor{blue}{$23.34$} & \textcolor{blue}{$0.7271$} & \textcolor{blue}{$0.1726$} & \pmb{$38.10$} & \pmb{$0.9502$} & \pmb{$0.0459$} \\
SEI (Ours) & \pmb{$23.63$} & \pmb{$0.7380$} & \pmb{$0.1689$} & \textcolor{blue}{$34.05$} & $0.8828$ & \textcolor{blue}{$0.0908$} \\
CSS & $21.10$ & $0.5827$ & $0.4017$ & $30.13$ & $0.8747$ & $0.2493$ \\
BM3D & $23.16$ & $0.6894$ & $0.3170$ & $33.81$ & \textcolor{blue}{$0.9216$} & $0.1850$ \\
DIP & $20.34$ & $0.4913$ & $0.4450$ & $30.19$ & $0.7739$ & $0.2922$ \\
\midrule
Blurry images & $18.99$ & $0.3786$ & $0.4468$ & $27.79$ & $0.7147$ & $0.4250$ \\
\bottomrule
\end{tabular}
\end{table}

\section{Conclusion}
\label{section.conclusion}

In this work, we introduce scale-equivariant imaging (SEI), a new self-supervised method for image deblurring and super-resolution which recovers the missing high-frequency content in images corrupted by a blur filter or by a downsampling operator directly from the corrupted images and without using a single clean image. Throughout extensive experiments, we show that SEI can be as effective as fully-supervised methods and that it outperforms other state-of-the-art self-supervised methods, including diffusion and plug-and-play approaches. This is consistent both across different degradation models and across different image distributions, including natural and medical images.

Unlike previous works which focus on rotations and translations~\cite{chen_equivariant_2021, chen_robust_2022}, we show that the right symmetry to exploit for bandlimited imaging problems is the invariance to scaling transforms as it enables the recovery of high-frequency information lost in the imaging process. On the contrary, exploiting rotations and translations cannot help to recover any information beyond the bandwidth of the corrupted images. Moreover, we present a novel approach to model identification theory~\cite{tachella_sensing_2023}, which, instead of recovering the original set of latent images, focuses on the recovery of a projection of this set.

The promising results of SEI suggest that it could be applied to the significantly more difficult blind imaging problems where the degradation model is only partially known. We expect that exploiting the symmetries of the degradation model will prove crucial to solve these problems.

% main content end %

\bibliography{manuscript}
\bibliographystyle{IEEEtran}

% Appendix %
\appendix

\subsection{Additional ablation studies} \label{section.additional_ablation_studies}

In this section, we present additional ablations made on the three different imaging modalities considered in the main paper. We study the influence of six additional hyperparameters on the final performance and report the results in~\Cref{tab:additional_ablation_studies}.

\textbf{1) Trade-off parameter.} Different values of $\alpha$ control the relative importance of the two terms in~\cref{eq.sei}. We compare the value $\alpha = 1$ used in the main experiments with $\alpha = 0.1$ and $\alpha = 10$ and observe that $\alpha = 1$ performs the best. It might indicate that the two terms are properly balanced without the need for the addition of a weighting factor, and motivates further theoretical analyses.

\textbf{2) Curriculum learning.}  We test if gradually increasing the trade-off parameter $\alpha$ throughout training is beneficial compared to using a fixed value. We increase $\alpha$ linearly from $0$ to $1$ every iteration. Our results show a performance drop compared to using a fixed value which might indicate that high values of $\alpha$ are needed early in the training to guide the network.

\textbf{3) Scale distribution.} We compare different scale distributions for computing the loss. We choose uniform scales on $\{ 1/2, 3/4 \}$ as in the main experiments, on $\{ 1/4, 7/8 \}$ to test scales closer to $0$ and $1$, and on the continuous intervals $[1/2, 3/4]$ and $[1/4, 7/8]$ to test the importance of intermediate scales. Our results show no clear benefit of using scales closer to $0$ and $1$ or intermediate scales. It might indicate that the network learns to reconstruct higher and higher frequencies as the training progresses rendering the use of a broader range of scales less impactful.

\textbf{4) Number of transformations.} In the main experiments, we transform each batch element exactly once per iteration. In this experiment, we explore augmenting each batch with $2$ or $4$ transformed versions of each batch element to see if it improves the performance at the expense of a higher memory footprint. Our results show no clear benefit of doing so which might indicate that the gradient steps are similar for the different scales.

\textbf{5) Interpolation scheme.} In the main experiments, we use bicubic interpolation to implement the downscaling transformations. Here, we also consider bilinear and nearest neighbor interpolation. Our results are mixed across super-resolution and deblurring. For deblurring, bicubic interpolation seems to perform better followed by bilinear interpolation, whereas for super-resolution nearest neighbor interpolation performs the best followed by bicubic interpolation. This might indicate a possible drawback of the blurring effect of bicubic and bilinear interpolation against certain bandlimiting filters.

\textbf{6) Network architecture.} In addition to SwinIR, we consider UNet~\cite{ronneberger2015u}, DRUnet~\cite{zhang2021plug} and SCUnet~\cite{zhang2023practical} backbones using the implementations from DeepInverse~\cite{tachella2025deepinverse}. We observe that while SwinIR and DRUNet perform similarly, UNet and SCUNet performs worse which might indicate that these architectures require a larger amount of training data.

\begin{table*}[th]
\centering
\setlength{\tabcolsep}{2pt} % default: 6pt
\renewcommand{\arraystretch}{1} % default: 1
\scriptsize
\caption{Additional ablation studies.}
\label{tab:additional_ablation_studies}
\begin{adjustwidth}{-10mm}{-10mm}
\resizebox{\linewidth}{!}{
  \begin{tabular}{>{\centering\arraybackslash}p{14mm}ccc ccc ccc}
\toprule
& \multicolumn{3}{c}{Gaussian filter (medium)} & \multicolumn{3}{c}{Box filter (medium)} & \multicolumn{3}{c}{Downsampling ($\times 2$)} \\
\cmidrule(lr){2-4} \cmidrule(lr){5-7} \cmidrule(lr){8-10}
$\alpha$ & PSNR $\uparrow$ & SSIM $\uparrow$ & LPIPS $\downarrow$ & PSNR $\uparrow$ & SSIM $\uparrow$ & LPIPS $\downarrow$ & PSNR $\uparrow$ & SSIM $\uparrow$ & LPIPS $\downarrow$ \\
\midrule
0.1 & $24.72 \pm 2.41$ & $0.6941 \pm 0.0542$ & $0.3983 \pm 0.0717$ & $25.82 \pm 2.64$ & $0.7604 \pm 0.0595$ & $0.2846 \pm 0.0789$ & \textcolor{blue}{$27.66 \pm 2.66$} & $0.8352 \pm 0.0357$ & \textcolor{blue}{$0.2245 \pm 0.0572$} \\
1 & \pmb{$25.86 \pm 2.75$} & \pmb{$0.7945 \pm 0.0813$} & \pmb{$0.2540 \pm 0.1074$} & \pmb{$27.04 \pm 2.95$} & \pmb{$0.8182 \pm 0.0744$} & \pmb{$0.1771 \pm 0.0832$} & \pmb{$28.09 \pm 3.03$} & \pmb{$0.8616 \pm 0.0568$} & \pmb{$0.1908 \pm 0.0765$} \\
10 & \textcolor{blue}{$25.57 \pm 2.70$} & \textcolor{blue}{$0.7748 \pm 0.0735$} & \textcolor{blue}{$0.3048 \pm 0.0921$} & \textcolor{blue}{$26.89 \pm 2.80$} & \textcolor{blue}{$0.8049 \pm 0.0659$} & \textcolor{blue}{$0.1945 \pm 0.0756$} & $27.21 \pm 2.63$ & \textcolor{blue}{$0.8364 \pm 0.0345$} & $0.3070 \pm 0.0678$ \\
\midrule
Curr.~lear. & PSNR $\uparrow$ & SSIM $\uparrow$ & LPIPS $\downarrow$ & PSNR $\uparrow$ & SSIM $\uparrow$ & LPIPS $\downarrow$ & PSNR $\uparrow$ & SSIM $\uparrow$ & LPIPS $\downarrow$ \\
\midrule
\texttimes & \pmb{$25.86 \pm 2.75$} & \pmb{$0.7945 \pm 0.0813$} & \pmb{$0.2540 \pm 0.1074$} & \pmb{$27.04 \pm 2.95$} & \pmb{$0.8182 \pm 0.0744$} & \pmb{$0.1771 \pm 0.0832$} & \pmb{$28.09 \pm 3.03$} & \pmb{$0.8616 \pm 0.0568$} & \pmb{$0.1908 \pm 0.0765$} \\
\checkmark &
$24.96 \pm 2.54$ & $0.7267 \pm 0.0611$ & $0.3622 \pm 0.0721$
& $26.46 \pm 2.75$ & $0.7892 \pm 0.0650$ & $0.2243 \pm 0.0766$
& $28.04 \pm 2.68$ & $0.8354 \pm 0.0412$ & $0.2416 \pm 0.0659$ \\
\midrule
Scale~dist. & PSNR $\uparrow$ & SSIM $\uparrow$ & LPIPS $\downarrow$ & PSNR $\uparrow$ & SSIM $\uparrow$ & LPIPS $\downarrow$ & PSNR $\uparrow$ & SSIM $\uparrow$ & LPIPS $\downarrow$ \\
\midrule
$\{1/2, 3/4\}$ & \pmb{$25.86 \pm 2.75$} & \pmb{$0.7945 \pm 0.0813$} & \pmb{$0.2540 \pm 0.1074$} & \pmb{$27.04 \pm 2.95$} & \pmb{$0.8182 \pm 0.0744$} & \pmb{$0.1771 \pm 0.0832$} & $28.09 \pm 3.03$ & \pmb{$0.8616 \pm 0.0568$} & \pmb{$0.1908 \pm 0.0765$} \\
$[1/2, 3/4]$ & $25.11 \pm 2.61$ & $0.7439 \pm 0.0664$ & \textcolor{blue}{$0.3129 \pm 0.0729$} & \textcolor{blue}{$26.70 \pm 2.78$} & \textcolor{blue}{$0.7982 \pm 0.0663$} & $0.2102 \pm 0.0784$ & \pmb{$28.30 \pm 2.70$} & $0.8466 \pm 0.0347$ & $0.2087 \pm 0.0640$ \\
$\{1/4, 7/8\}$ & \textcolor{blue}{$25.52 \pm 2.70$} & \textcolor{blue}{$0.7651 \pm 0.0696$} & $0.3481 \pm 0.0906$ & $26.49 \pm 2.74$ & $0.7897 \pm 0.0637$ & $0.2177 \pm 0.0737$ & \textcolor{blue}{$28.27 \pm 2.68$} & $0.8423 \pm 0.0385$ & $0.2137 \pm 0.0643$ \\
$[1/4, 7/8]$ & $25.39 \pm 2.67$ & $0.7534 \pm 0.0700$ & $0.3419 \pm 0.0836$ & $26.64 \pm 2.76$ & $0.7957 \pm 0.0651$ & \textcolor{blue}{$0.2080 \pm 0.0752$} & $28.27 \pm 2.67$ & \textcolor{blue}{$0.8480 \pm 0.0357$} & \textcolor{blue}{$0.1937 \pm 0.0617$} \\
\midrule
Nb~trans. & PSNR $\uparrow$ & SSIM $\uparrow$ & LPIPS $\downarrow$ & PSNR $\uparrow$ & SSIM $\uparrow$ & LPIPS $\downarrow$ & PSNR $\uparrow$ & SSIM $\uparrow$ & LPIPS $\downarrow$ \\
\midrule
1 & \pmb{$25.86 \pm 2.75$} & \pmb{$0.7945 \pm 0.0813$} & \pmb{$0.2540 \pm 0.1074$} & \pmb{$27.04 \pm 2.95$} & \pmb{$0.8182 \pm 0.0744$} & \pmb{$0.1771 \pm 0.0832$} & $28.09 \pm 3.03$ & \pmb{$0.8616 \pm 0.0568$} & \pmb{$0.1908 \pm 0.0765$} \\
2 & \textcolor{blue}{$25.43 \pm 2.66$} & \textcolor{blue}{$0.7614 \pm 0.0704$} & \textcolor{blue}{$0.3301 \pm 0.0855$} & $26.77 \pm 2.79$ & $0.7997 \pm 0.0661$ & $0.2062 \pm 0.0767$ & \textcolor{blue}{$28.38 \pm 2.75$} & $0.8491 \pm 0.0373$ & $0.2060 \pm 0.0650$ \\
4 & $25.34 \pm 2.63$ & $0.7555 \pm 0.0672$ & $0.3415 \pm 0.0828$ & \textcolor{blue}{$26.82 \pm 2.79$} & \textcolor{blue}{$0.8020 \pm 0.0662$} & \textcolor{blue}{$0.2052 \pm 0.0775$} & \pmb{$28.44 \pm 2.75$} & \textcolor{blue}{$0.8510 \pm 0.0363$} & \textcolor{blue}{$0.1959 \pm 0.0616$} \\
\midrule
Interp. & PSNR $\uparrow$ & SSIM $\uparrow$ & LPIPS $\downarrow$ & PSNR $\uparrow$ & SSIM $\uparrow$ & LPIPS $\downarrow$ & PSNR $\uparrow$ & SSIM $\uparrow$ & LPIPS $\downarrow$ \\
\midrule
Bicubic & \pmb{$25.86 \pm 2.75$} & \pmb{$0.7945 \pm 0.0813$} & \pmb{$0.2540 \pm 0.1074$} & \pmb{$27.04 \pm 2.95$} & \pmb{$0.8182 \pm 0.0744$} & \pmb{$0.1771 \pm 0.0832$} & $28.09 \pm 3.03$ & \pmb{$0.8616 \pm 0.0568$} & \pmb{$0.1908 \pm 0.0765$} \\
Bilinear & \textcolor{blue}{$25.50 \pm 2.71$} & \textcolor{blue}{$0.7642 \pm 0.0719$} & \textcolor{blue}{$0.3280 \pm 0.0893$} & \textcolor{blue}{$26.76 \pm 2.81$} & \textcolor{blue}{$0.8006 \pm 0.0671$} & $0.2191 \pm 0.0803$ & \pmb{$28.42 \pm 2.77$} & \textcolor{blue}{$0.8500 \pm 0.0400$} & $0.2070 \pm 0.0635$ \\
Nearest & $25.10 \pm 2.62$ & $0.7392 \pm 0.0648$ & $0.3806 \pm 0.0820$ & $26.60 \pm 2.77$ & $0.7941 \pm 0.0661$ & \textcolor{blue}{$0.2184 \pm 0.0782$} & \textcolor{blue}{$28.40 \pm 2.75$} & $0.8485 \pm 0.0398$ & \textcolor{blue}{$0.2015 \pm 0.0633$} \\
\midrule
Arch. & PSNR $\uparrow$ & SSIM $\uparrow$ & LPIPS $\downarrow$ & PSNR $\uparrow$ & SSIM $\uparrow$ & LPIPS $\downarrow$ & PSNR $\uparrow$ & SSIM $\uparrow$ & LPIPS $\downarrow$ \\
\midrule
SwinIR & \pmb{$25.86 \pm 2.75$} & \pmb{$0.7945 \pm 0.0813$} & \pmb{$0.2540 \pm 0.1074$} & \pmb{$27.04 \pm 2.95$} & \pmb{$0.8182 \pm 0.0744$} & \pmb{$0.1771 \pm 0.0832$} & $28.09 \pm 3.03$ & \pmb{$0.8616 \pm 0.0568$} & \pmb{$0.1908 \pm 0.0765$} \\
UNet & $24.56 \pm 2.49$ & $0.6690 \pm 0.0557$ & $0.3862 \pm 0.0652$ & $26.16 \pm 2.76$ & $0.7735 \pm 0.0695$ & \textcolor{blue}{$0.2217 \pm 0.0715$} & \textcolor{blue}{$28.20 \pm 2.70$} & $0.8443 \pm 0.0360$ & $0.2224 \pm 0.0641$ \\
DRUNet & \textcolor{blue}{$25.65 \pm 2.70$} & \textcolor{blue}{$0.7405 \pm 0.0647$} & \textcolor{blue}{$0.3856 \pm 0.0820$} & \textcolor{blue}{$26.28 \pm 2.72$} & \textcolor{blue}{$0.7760 \pm 0.0653$} & $0.2477 \pm 0.0756$ & \pmb{$28.48 \pm 2.80$} & \textcolor{blue}{$0.8581 \pm 0.0372$} & \textcolor{blue}{$0.1964 \pm 0.0602$} \\
SCUNet & $22.77 \pm 1.95$ & $0.6607 \pm 0.0590$ & $0.4240 \pm 0.0647$ & $24.09 \pm 2.03$ & $0.7370 \pm 0.0516$ & $0.3207 \pm 0.0808$ & $27.51 \pm 2.39$ & $0.8393 \pm 0.0360$ & $0.2143 \pm 0.0625$ \\
\bottomrule
\end{tabular}
}
\end{adjustwidth}
\end{table*}

\subsection{Proofs}

We prove the two theorems introduced in \Cref{section.theory} and define the Fourier transform and the convolution operation.

\begin{definition}[Fourier transform, inverse Fourier transform]
	Let $z \in \mathcal S$, we consider the Fourier transform of $z$ to be the function $\hat z \in \mathcal S$ defined by
	\begin{equation}
		\hat z(\xi) = \int_{\mathbb R^2} z(\varpos) e^{-i2\pi \xi^\top \varpos} \mathrm d\varpos,\ \xi \in \mathbb R^2.
	\end{equation}
	Conversely, for $\hat z \in \mathcal S$, the inverse Fourier transform of $\hat z$ is the function $z \in \mathcal S$ defined by
	\begin{equation}
		z(\varpos) = \int_{\mathbb R^2} \hat z(\xi) e^{i2\pi \xi^\top \varpos} \mathrm d\xi,\ \varpos \in \mathbb R^2.
	\end{equation}
\end{definition}

\begin{definition}[Convolution]
	Let $\varkernel, z \in \mathcal S$, we consider the convolution of $\varkernel$ and $z$ to be the function $\varkernel * z \in \mathcal S$ defined by
	\begin{equation}
		(\varkernel * z)(\varpos) = \int_{\mathbb R^2} \varkernel(\varposalt) z(\varpos - \varposalt) \mathrm d\varposalt,\ \varpos \in \mathbb R^2.
	\end{equation}
\end{definition}

% Restatement of the theorem
\theoremSpecialEuclidean*

\begin{proof}
Let's assume that $\varkernel$ and $\varphi$ are bandlimited with $\xi_\varkernel < \xi_\varphi$. Let
\begin{equation}
	\mathcal Z_1 = \{ z \in \mathcal S, \hat z(\xi) = 0, \forall \lVert \xi \rVert_2 \geq \xi_\varkernel \}
\end{equation}
be the set of images with bandwidth at most $\xi_\varkernel$ and $\mathcal Z_2 = \mathcal S$ the set of all images (bandlimited or not). We show that 1) $\mathcal Z_1$ and $\mathcal Z_2$ are invariant under roto-translations, 2) that $\mathcal Y_1 = \mathcal Y_2$, and 3) that $\mathcal X_1 \neq \mathcal X_2$. 1) The subset $\mathcal Z_2$ is the whole image space $\mathcal S$ which is invariant under roto-translations, and since the bandwidth of an image does not change under roto-translations, the subset $\mathcal Z_1$ is also invariant under roto-translations. 2) The subsets $\mathcal Z_1$ and $\mathcal Z_2$ contain every image with bandwidth at most $\xi_\varkernel$ so $\mathcal Y_1$ and $\mathcal Y_2$ are both equal to the set of all images with bandwidth at most $\xi_\varkernel$ which also vanish when the filter vanishes \begin{equation}
	\mathcal Y_1 = \mathcal Y_2 = \{ y \in \mathcal S, \hat \varkernel(\xi) = 0 \implies \hat y(\xi) = 0,\ \forall \xi \in \mathbb R^2 \}.
\end{equation}
3) Consider the image $x = \varphi * \varphi$, it has bandwidth $\xi_\varphi$ and it is of the form $\varphi * z$ for $z \in \mathcal S$, so $x \in \mathcal X_2$ but $x \notin \mathcal X_1$ as $\mathcal X_1$ contains only images with bandwidth at most $\xi_\varkernel$ and $\xi_\varkernel < \xi_\varphi$.
\end{proof}

% Restatement of the theorem
\theoremHomothety*

\begin{proof}
	Let $\xi \in \mathbb R^2$. Since $\hat \varkernel(0) \neq 0$, and since $\hat \varkernel$ is continuous as it is in $\mathcal S$, there is a $\xi_h > R_\varkernel > 0$ such that
	\begin{equation}
		\hat \varkernel(\xi) \neq 0, \forall \xi \in \mathbb R^2, \lVert \xi \rVert_2 \leq R_\varkernel.
	\end{equation}
	Let $s = \frac{\xi_\varphi}{R_\varkernel}$ and
	\begin{equation}
		\hat {\mathcal X} = \{ x_y, y \in \mathcal Y \}.
	\end{equation}
	Notice that $x_y$ as defined in \cref{eq:def:xy} is only well-defined if $\hat \varkernel(s^{-1}\xi) \neq 0$ for all $\lVert \xi \rVert_2 < \xi_\varphi$, or equivalently if $\hat \varkernel(\xi) \neq 0$ for all $\lVert \xi \rVert < s^{-1} \xi_{\varphi} = R_\varkernel$, which is the case for that specific value of $s$. We are going to show that $\hat {\mathcal X} = \mathcal X$ where
	\begin{equation}
		\mathcal X = \{ \varphi * z,\ z \in \mathcal Z \}.
	\end{equation}
	First, we show that $\hat {\mathcal X} \subseteq \mathcal X$.
	For $y \in \mathcal Y$, there is a $z \in \mathcal Z$ such that
	\begin{equation}
		y = \varkernel * z,
	\end{equation}
	and according to the convolution theorem,
	\begin{equation} \label{eq:yhz:conv}
		\hat y(s^{-1} \xi) = \varkernel(s^{-1} \xi)\hat z(s^{-1} \xi).
	\end{equation}
	The $\hat x_y$ associated with this $y$ is obtained by replacing \cref{eq:yhz:conv} in \cref{eq:def:xy}
	\begin{equation} \label{eq:xy2}
		\hat x_y(\xi) = \begin{cases}
			s^ {-1} \hat \varphi(\xi)\hat z(s^ {-1}\xi), & \text{if $\lVert \xi \rVert_2 < \xi_\varphi$}, \\
			0, & \text{otherwise}.
		\end{cases}
	\end{equation}
	Now, let $z_s = \Sigma_{s^{-1}} z$, according to the scaling formula for the Fourier transform,
	\begin{equation}
		\hat z_s(\xi) = s^{-1} \hat z(s^{-1} \xi).
	\end{equation}
	Since $z \in \mathcal Z$ and since $\mathcal Z$ is invariant to scale, $z_s \in \mathcal Z$ and $x \in \mathcal X$ for
	\begin{equation}
		x = \varphi * z_s.
	\end{equation}
	Moreover, according to the convolution theorem,
	\begin{equation}
		\hat x(\xi) = \hat \varphi(\xi) \hat z_s(\xi),
	\end{equation}
	and more precisely since $\varphi$ is bandlimited
	\begin{equation} \label{eq:xsphiz}
		\hat x(\xi) = \begin{cases}
			\hat \varphi(\xi) \hat z_s(\xi), & \text{if $\lVert \xi \rVert_2 < \xi_\varphi$}, \\
			0, & \text{otherwise}.
		\end{cases}
	\end{equation}
	Comparing \cref{eq:xy2} with \cref{eq:xsphiz} we have
	\begin{equation}
		\hat x_y(\xi) = \hat x(\xi),
	\end{equation}
	thus $x_y = x$, and since $x \in \mathcal X$, $\hat {\mathcal X} \subseteq \mathcal X$. \\
	Conversely, we show that $\mathcal X \subseteq \hat {\mathcal X}$. Let $x \in \mathcal X$, then there is a $z \in \mathcal Z$ such that
	\begin{equation}
		x = \varphi * z.
	\end{equation}
	Now, let
	\begin{equation}
		y_s = \varkernel * (\mathop{\Sigma_s} z).
	\end{equation}
	According to the convolution theorem,
	\begin{equation}
		\hat x(\xi) = \hat \varphi(\xi) \hat z(\xi),
	\end{equation}
	and
	\begin{align}
		\hat y_s(\xi) = \hat \varkernel(\xi) \hat z_s(\xi),
	\end{align}
	where $z_s = \mathop{\Sigma_s} z$, and according to the scaling formula for the Fourier transform,
	\begin{equation}
		\hat z_s(\xi) = s\hat z(s\xi) = s (\mathop{\Sigma_{s^{-1}}} \hat z)(\xi),
	\end{equation}
	as
	\begin{equation}
		z_s(\varpos) = z(s^{-1}\varpos), \forall \varpos \in \mathbb R^2.
	\end{equation}
	If $\lVert \xi \rVert_2 < \xi_\varphi$, then it follows directly from the previous equations that
	\begin{align}
		&\hat x(\xi) = \hat \varphi(\xi) \hat z(\xi), \\
		&= \hat \varkernel(s^{-1}\xi)^{-1} \hat \varphi(\xi) \hat \varkernel(s^{-1}\xi)\hat z(\xi), \\
		&= \hat \varkernel(s^{-1}\xi)^{-1} \hat \varphi(\xi) \Big( \hat \varkernel(s^{-1}\xi)\hat z(ss^{-1}\xi)\Big), \\
		&= \hat \varkernel(s^{-1}\xi)^{-1} \hat \varphi(\xi) \Big( \hat \varkernel(s^{-1}\xi)(\mathop{\Sigma_{s^{-1}}} \hat z)(s^{-1}\xi)\Big), \\
		&= \hat \varkernel(s^{-1}\xi)^{-1} \hat \varphi(\xi) \Big( \hat \varkernel(s^{-1}\xi)(s^{-1}s\mathop{\Sigma_{s^{-1}}} \hat z)(s^{-1}\xi)\Big), \\
		&= s^{-1} \hat \varkernel(s^{-1}\xi)^{-1} \hat \varphi(\xi) \Big( \hat \varkernel(s^{-1}\xi)\hat z_s(s^{-1}\xi)\Big), \\
		&= s^{-1} \hat \varkernel(s^{-1}\xi)^{-1} \hat \varphi(\xi) \hat y_s(s^{-1}\xi).
	\end{align}
	Hence, as $\varphi$ is bandlimited,
	\begin{equation}
		\hat x(\xi) = \begin{cases}
			\frac{\hat \varphi(\xi) \hat y_s(s^{-1}\xi)}{s \hat \varkernel(s^{-1}\xi)} &\text{if $\lVert \xi \rVert_2 < \xi_\varphi$,} \\
			0 &\text{otherwise,}
		\end{cases}
	\end{equation}
	thus $x = x_y$ and since $y \in \mathcal Y$, $x \in \hat {\mathcal X}$, thus $\mathcal X \subseteq \hat {\mathcal X}$.
\end{proof}

\end{document}